\documentclass[11pt,twoside]{article}
\usepackage{amssymb,latexsym,graphicx,hyperref}
\usepackage{epstopdf}
\normalsize
\makeatletter
\@addtoreset{equation}{section}
\makeatother

\newtheorem{theorem}{Theorem}
\newtheorem{lemma}{Lemma}
\newenvironment{proof}{{\sc Proof. }}{\hfill$\Box$\vspace{0.2in}}
\title{An approximation algorithm for the Bandpass-$2$ problem}
\author{Weitian Tong\thanks{Department of Computing Science, University of Alberta.
Edmonton, Alberta T6G 2E8, Canada.}$~^,$\thanks{Email: {\tt weitian@ualberta.ca}}
\and
Zhi-Zhong Chen\thanks{Division of Information System Design, Tokyo Denki University.
Hatoyama, Saitama 350-0394, Japan.
Email: {\tt zzchen@mail.dendai.ac.jp}}
\and
Lusheng Wang\thanks{Department of Computer Science, City University of Hong Kong.
Kowloon, Hong Kong, China.
Email: {\tt cswangl@cityu.edu.hk}}
\and
Yinfeng Xu\thanks{Business School, Sichuan University.
Chengdu, Sichuan 610065, China.}$~^,$\thanks{Email: {\tt yfxu@scu.edu.cn}}
\and
Jiuping Xu$^\P$$^,$\thanks{Email: {\tt xujiuping@scu.edu.cn}}
\and
Randy Goebel$^*$$^,$\thanks{Email: {\tt rgoebel@ualberta.ca}}
\and
Guohui Lin$^*$$^,$\thanks{Correspondence author.  Email: {\tt guohui@ualberta.ca}}}
\date{\today}
\begin{document}
\maketitle
\begin{abstract}
The general Bandpass-$B$ problem is NP-hard and can be approximated by a reduction into the weighted $B$-set packing problem,
with a worst case performance ratio of $O(B^2)$.
When $B = 2$, a maximum weight matching gives a $2$-approximation to the problem.
In this paper, we call the Bandpass-$2$ problem simply the Bandpass problem.
The Bandpass problem can be viewed as a variation of the maximum traveling salesman problem,
in which the edge weights are dynamic rather than given at the front.
We present a $\frac {426}{227}$-approximation algorithm for the problem.
Such an improved approximation is built on an intrinsic structural property proven for the optimal solution and
several novel schemes to partition a $b$-matching into desired matchings.

\paragraph{Keywords:}
The Bandpass problem, maximum weight $b$-matching, acyclic $2$-matching, approximation algorithm, worst case performance ratio
\end{abstract}

\section{Introduction}
In optical communication networks,
a sending point uses a binary matrix $A_{n \times m}$ to send $n$ information packages to $m$ different destination points,
in which the entry $a_{ij} = 1$ if information package $i$ is {\em not} destined for point $j$, or $a_{ij} = 0$ otherwise.
To achieve the highest cost reduction via wavelength division multiplexing technology,
an optimal packing of information flows on different wavelengths into groups is necessary~\cite{BBN09}. 
Under this binary matrix representation,
every $B$ consecutive $1$'s in a column indicates an opportunity for merging information to reduce the communication cost,
where $B$ is a pre-specified positive integer called the {\em bandpass number}.
Such a set of $B$ consecutive $1$'s in a column of the matrix is said to form a {\em bandpass}.
When counting the number of bandpasses in the matrix, no two of them in the same column are allowed to share any common rows.
The computational problem, the {\em Bandpass-$B$ problem},
is to find an optimal permutation of rows of the input matrix $A_{n \times m}$
such that the total number of extracted bandpasses in the resultant matrix is maximized~\cite{BB04,BBN09,Lin11}.
Note that though multiple bandpass numbers can be used in practice, for the sake of complexities and costs,
usually only one fixed bandpass number is considered~\cite{BBN09}.

The general Bandpass-$B$ problem, for any fixed $B \ge 2$, has been proven to be NP-hard~\cite{Lin11}.
In fact, the NP-hardness of the Bandpass-$2$ problem can be proven by a reduction from the well-known {\em Hamiltonian path} problem~\cite{GJ79},
where in the constructed binary matrix $A_{n \times m}$, a row maps to a vertex, a column maps to an edge,
and $a_{ij} = 1$ if and only if edge $e_j$ is incident to vertex $v_i$.
It follows that there is a row permutation achieving $n-1$ bandpasses if and only if there is a Hamiltonian path in the graph.

On the approximability, the Bandpass-$B$ problem has a close connection to the weighted $B$-set packing problem~\cite{GJ79}.
Given an instance $I$ of a maximization problem $\Pi$,
let $C^*(I)$ ($C(I)$, respectively) denote the value of the optimal solution (the value of the solution produced by an algorithm, respectively).
The performance ratio of the algorithm on $I$ is $\frac {C^*(I)}{C(I)}$.
The algorithm is a $\rho$-approximation if $\sup_I \frac {C^*(I)}{C(I)} \le \rho$.
By taking advantages of the approximation algorithms designed for the weighted $B$-set packing problem~\cite{AH98,CH99},
the Bandpass-$B$ problem can be approximated within $O(B^2)$~\cite{Lin11}.
Moreover, since the maximum weight matching problem is solvable in cubic time,
the Bandpass-$2$ problem admits a simple maximum weight matching based $2$-approximation algorithm~\cite{Lin11}.
In the sequel, we call the Bandpass-$2$ problem simply the Bandpass problem.

In a preliminary version of this paper~\cite{TGD12},
Tong {\it et al.} presented a first improved approximation algorithm for the Bandpass problem using two maximum weight matchings.
Their algorithm has a worst case performance ratio of $\frac {36}{19} \approx 1.8948$.
In more details, their algorithm computes a maximum weight matching in the edge-weighted graph constructed from the input matrix,
and then computes another maximum weight matching in the residual graph;
Through breaking cycles in the union of these two matchings, a collection of paths are achieved and concatenated into a Hamiltonian path.
When estimating the weight of this Hamiltonian path, the authors present an intrinsic structural property for the optimal row permutation,
and show that the weight of the second maximum weight matching is greater than or equal to a portion of certain bandpasses in the optima.
These special bandpasses in the optima are characterized using the first maximum weight matching.

Instead of breaking cycles in the union of two matchings,
Chen and Wang~\cite{CW12} presented an alternative to compute a second matching such that the union of the two matchings is guaranteed acyclic.
The keys to this success are two lemmas that partition a $2$-matching (an acyclic $2$-matching, respectively)
into $4$ ($3$, respectively) candidate matchings.
Still based on the structural property characterized in~\cite{TGD12},
this alternative algorithm achieves a better performance ratio of $\frac {220}{117} \approx 1.8805$.

In this paper, we advance the novel $2$-matching partitioning scheme to show that two edge disjoint $2$-matchings
can be partitioned into $7.5$ desired matchings such that the union of each of them and the very first maximum weight matching is guaranteed acyclic.
Including here all the detailed proofs missed from the two preliminary versions~\cite{TGD12,CW12},
we show that our new approximation algorithm has a performance ratio of $\frac {426}{227} \approx 1.8767$.

\section{The approximation algorithm}
A reduction from the Hamiltonian path problem has been used to prove the NP-hardness of the Bandpass problem.
But the Bandpass problem does not readily reduce to the maximum traveling salesman problem (Max-TSP) \cite{GJ79} for approximation algorithm design.
The main reason is that, an instance graph of Max-TSP is {\em static}, in that all (non-negative) edge weights are given at the front,
while in the Bandpass problem the number of bandpasses extracted between two consecutive rows in a row permutation is permutation dependent.
Nevertheless, as shown in the sequel, our design idea is based on maximum weight $b$-matchings for $b = 1, 2$, and $4$,
similarly as in approximating Max-TSP~\cite{Ser84,HR00,COW05,PMM09}.
Formally, in Max-TSP, a complete edge-weighted graph is given, where the edge weights are non-negative integers,
and the goal is to compute a Hamiltonian cycle with the maximum weight.
Note that there are several variants of Max-TSP been studied in the literature.
In our case, the input graph is undirected (or symmetric) and the edge weights do not necessarily satisfy the triangle inequality.
The following Lemma~\ref{lemma1} states the currently best approximation result for Max-TSP.

\begin{lemma}
\label{lemma1}{\rm \cite{PMM09}}
The Max-TSP admits an $O(n^3)$-time $\frac 97$-approximation algorithm, where $n$ is the number of vertices in the graph.
\end{lemma}

In our Bandpass problem, since we can always add a row of all $0$'s if needed,
we assume without loss of generality that the number of rows, $n$, is even.
A {\em $b$-matching} of a graph is a subgraph in which the degree of each vertex is at most $b$.
A maximum weight $b$-matching of an edge weighted graph can be computed in $O(n^2m)$ time~\cite{Gab83,Ans87,MP95},
where $n$ is the number of vertices and $m$ is the number of edges in the graph.
Note that a $2$-matching is a collection of vertex-disjoint cycles and paths.
A $2$-matching is {\em acyclic} if it does not contain any cycle ({\it i.e.}, it is a collection of vertex-disjoint paths).

Given the input binary matrix $A_{n \times m}$, let $r_i$ denote the $i$-th row.
We first construct a graph $G$ of which the vertex set is exactly the row set $\{r_1, r_2, \ldots, r_n\}$.
Between rows $r_i$ and $r_j$, the {\em static} edge weight is defined as the maximum number of bandpasses that can be formed between the two rows,
and is denoted as $w(i, j)$.
In the sequel we use row (of the matrix) and vertex (of the graph) interchangeably.

For a row permutation $\pi = (\pi_1, \pi_2, \ldots, \pi_n)$, its $i$-th row is the $\pi_i$-th row in the input matrix.
We call a maximal segment of consecutive $1$'s in a column of $\pi$ a {\em strip} of $\pi$.
The length of a strip is defined to be the number of $1$'s therein.
A length-$\ell$ strip contributes exactly $\lfloor \frac {\ell}2 \rfloor$ bandpasses to the permutation $\pi$.
We use $S_\ell(\pi)$ to denote the set of all length-$\ell$ strips of $\pi$, and $s_\ell(\pi) = |S_\ell(\pi)|$.
Let $b(\pi)$ denote the number of bandpasses extracted from the permutation $\pi$.
We have
\begin{equation}
\label{eq1}
b(\pi) = \sum_{\ell = 2}^n s_\ell(\pi) \left\lfloor \frac {\ell}2 \right\rfloor
= s_2(\pi) + \sum_{\ell = 3}^n s_\ell(\pi) \left\lfloor \frac {\ell}2 \right\rfloor.
\end{equation}
Let $p(\pi)$ denote the number of pairs of consecutive $1$'s in the permutation $\pi$.
A length-$\ell$ strip contributes exactly $\ell - 1$ pairs to the permutation $\pi$.
We have
\begin{equation}
\label{eq2}
p(\pi) = \sum_{\ell = 2}^n s_\ell(\pi) (\ell - 1) = s_2(\pi) + \sum_{\ell = 3}^n s_\ell(\pi) (\ell - 1).
\end{equation}

\subsection{Algorithm description}
In our algorithm denoted as {\sc Approx}, the first step is to compute a maximum weight matching $M_1$ in graph $G$.
Recall that there are an even number of rows.
Therefore, $M_1$ is a perfect matching (even though some edge weights could be $0$).
Let $w(M_1)$ denote the sum of its edge weights, indicating that exactly $w(M_1)$ bandpasses can be extracted from the row pairings suggested by $M_1$.
These bandpasses are called the bandpasses of $M_1$.

Next, every $1$ involved in a bandpass of $M_1$ is changed to $0$.
Let the resultant matrix be denoted as $A'_{n \times m}$,
the resultant edge weight between rows $r_i$ and $r_j$ be $w'(i, j)$
--- which is the maximum number of bandpasses can be formed between the two revised rows,
and the corresponding residual graph be denoted as $G'$.
One can see that if an edge $(r_i, r_j)$ belongs to $M_1$, then the new edge weight $w'(i, j) = 0$.
In the second step of {\sc Approx}, we compute a maximum weight $4$-matching ${\cal C}$ in graph $G'$,
which is further decomposed in $O(n^{2.5})$ time into two $2$-matchings denoted as ${\cal C}_1$ and ${\cal C}_2$~\cite{Har69,Die05}.
Let $w'({\cal C})$ denote the weight (the number of bandpasses) of ${\cal C}$ in the residual graph $G'$.
It is noted that no bandpass of ${\cal C}$ shares a $1$ with any bandpass of $M_1$.
Using $M_1$ and ${\cal C}_1$ and ${\cal C}_2$, by Lemma~\ref{lemma41},
we can compute a matching $M_2$ from ${\cal C}$ of weight at least $\frac 1{7.5} w'({\cal C})$ such that $G[M_1 \cup M_2]$ is guaranteed acyclic.

In the third step, we use the $\frac 97$-approximation algorithm described in~\cite{PMM09} to compute a Hamiltonian path ${\cal P}$ in $G'$
whose weight is at least $\frac 79$ of the maximum weight of a Hamiltonian path.
Then, using $M_1$ and ${\cal P}$, by Lemma~\ref{lem:mine},
we can compute another matching $M_2$ from ${\cal P}$ of weight at least $\frac 13 w'({\cal P})$ such that $G[M_1 \cup M_2]$ is guaranteed acyclic.

In the last step, we choose the larger one between the two $M_2$'s found in the last two steps,
and arbitrarily stack the paths in $G[M_1 \cup M_2]$ to give a row permutation $\pi$.
Note that the number of bandpasses extracted from $\pi$, $b(\pi)$, is greater than or equal to $w(M_1) + w'(M_2)$.

\subsection{Performance analysis}
Let $\pi^*$ denote the optimal row permutation such that its $b(\pi^*)$ is maximized over all row permutations.
Correspondingly, $S_2(\pi^*)$ denotes the set of length-$2$ strips in $\pi^*$, which contributes exactly $s_2(\pi^*)$ bandpasses towards $b(\pi^*)$.
The key part in the performance analysis for algorithm {\sc Approx} is to estimate $w'(M_2)$, as done in the following.

First, we partition the bandpasses of $S_2(\pi^*)$ into four groups: $B_1$, $B_2$, $B_3$, $B_4$.
Note that bandpasses of $S_2(\pi^*)$ do not share any $1$ each other.
$B_1$ consists of the bandpasses of $S_2(\pi^*)$ that also belong to matching $M_1$ (such as the one between rows $r_a$ and $r_b$ in Figure~\ref{s2});
$B_2$ consists of the bandpasses of $S_2(\pi^*)$ such that they are uniquely paired up to contribute a $1$ each to form a bandpass of $M_1$
(the other $1$ in each bandpass of $B_2$ is thus not shared by any other bandpass of $M_1$);
$B_3$ consists of the bandpasses of $S_2(\pi^*)$, each of which shares a $1$ with at least one bandpass of $M_1$,
and if it shares a $1$ with only one bandpass of $M_1$ then the other $1$ of this bandpass of $M_1$ is not shared by any other bandpass of $S_2(\pi^*)$;
$B_4$ consists of the remaining bandpasses of $S_2(\pi^*)$.
Figure~\ref{s2} illustrates some examples of these bandpasses, where bandpasses of $S_2(\pi^*)$ are in ovals and bandpasses of $M_1$ are in boxes.
\begin{figure}[htb]
\begin{center}
\unitlength=0.4mm
\begin{picture}(220, 130)
\put(50, 10){$\vdots$}
\put(0, 30){$a:$}
\put(20, 35){\line(1, 0){200}}
\put(100, 30){$1$}
\put(0, 40){$b:$}
\put(20, 45){\line(1, 0){200}}
\put(100, 40){$1$}
\put(98, 28){\framebox(6, 22){}}
\put(101, 39){\oval(11, 18){}}
\put(111, 30){$B_1$}
\put(50, 50){$\vdots$}
\put(0, 60){$t:$}
\put(20, 65){\line(1, 0){200}}
\put(100, 60){$0$}
\put(0, 70){$i:$}
\put(20, 75){\line(1, 0){200}}
\put(100, 70){$1$}
\put(0, 80){$j:$}
\put(20, 85){\line(1, 0){200}}
\put(100, 80){$1$}
\put(98, 78){\framebox(6, 22){}}
\put(101, 79){\oval(11, 18){}}
\put(0, 90){$k:$}
\put(20, 95){\line(1, 0){200}}
\put(100, 90){$1$}
\put(0, 100){$\ell:$}
\put(20, 105){\line(1, 0){200}}
\put(100, 100){$1$}
\put(101, 99){\oval(11, 18){}}
\put(0, 110){$u:$}
\put(20, 115){\line(1, 0){200}}
\put(100, 110){$0$}
\put(130, 110){$0$}
\put(30, 100){$1$}
\put(30, 110){$1$}
\put(28, 98){\framebox(6, 22){}}
\put(160, 60){$1$}
\put(160, 70){$1$}
\put(158, 58){\framebox(6, 22){}}
\put(141, 97){$B_3$}
\put(130, 70){$0$}
\put(130, 80){$1$}
\put(128, 78){\framebox(6, 22){}}
\put(130, 90){$1$}
\put(130, 100){$1$}
\put(131, 99){\oval(11, 18){}}
\put(81, 80){$B_2$}
\put(180, 110){$0$}
\put(180, 80){$0$}
\put(180, 100){$1$}
\put(180, 90){$1$}
\put(182, 98){\oval(11, 18){}}
\put(191, 97){$B_4$}
\put(50, 120){$\vdots$}
\end{picture}
\end{center}
\caption{An illustration of the bandpasses of $S_2(\pi^*)$ (in ovals) and the bandpasses of $M_1$ (in boxes) for grouping purpose.
	A horizontal line in the figure represents a row, led by its index.
    Rows that are adjacent in $\pi^*$ and/or row pairs of $M_1$ are intentionally ordered adjacently.
	In this figure, rows $r_a$ and $r_b$ are adjacent in $\pi^*$, denoted as $(r_a, r_b) \in \pi^*$, and edge $(r_a, r_b) \in M_1$ as well;
	the bandpasses between these two rows in $S_2(\pi^*)$ thus belong to $B_1$.
	Edges $(r_t, r_i), (r_j, r_k), (r_\ell, r_u) \in M_1$, while $(r_i, r_j), (r_k, r_\ell) \in \pi^*$;
	the bandpasses between rows $r_i$ and $r_j$ and between rows $r_k$ and $r_\ell$ in $S_2(\pi^*)$ shown in the figure
	have their group memberships indicated beside them respectively.\label{s2}}
\end{figure}
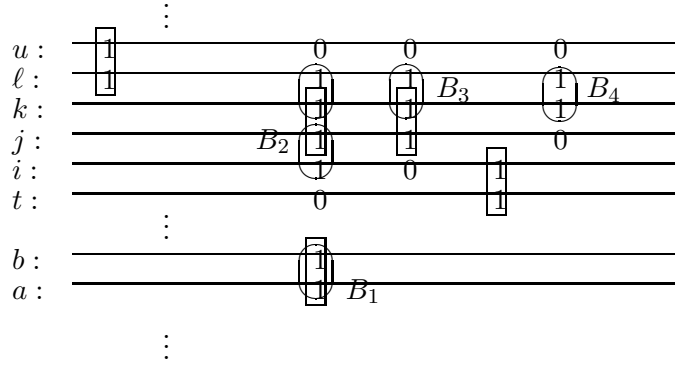

By the definition of partition, we have
\begin{equation}
\label{eq3}
s_2(\pi^*) = |B_1| + |B_2| + |B_3| + |B_4|.
\end{equation}
From these ``group'' definitions, we know all bandpasses of $B_1$ are in $M_1$.
Also, one pair of bandpasses of $B_2$ correspond to a distinct bandpass of $M_1$.
Bandpasses of $B_3$ can be further partitioned into subgroups such that
a subgroup of bandpasses together with a distinct maximal subset of bandpasses of $M_1$ form into an alternating cycle or path of length at least $2$.
Moreover,
1) when the path length is even, the number of bandpasses of this subgroup of $B_3$ is equal to the number of bandpasses of this subset of bandpasses of $M_1$;
2) when the path length is odd,
2a) either the number of bandpasses of this subgroup of $B_3$ is $1$ greater than the number of bandpasses of this subset of bandpasses of $M_1$,
2b) or the path length has to be at least $5$ and
so the number of bandpasses of this subgroup of $B_3$ is at least $\frac 23$ of the number of bandpasses of this subset of bandpasses of $M_1$.
It follows from 1), 2a) and 2b) that with respect to $B_3$, $M_1$ contains at least $\frac 23 |B_3|$ corresponding bandpasses.
That is,
\begin{equation}
\label{eq4}
w(M_1) \ge |B_1| + \frac 12 |B_2| + \frac 23 |B_3|.
\end{equation}
Clearly, all bandpasses of $B_4$ are in the residual graph $G'$,
while none of $B_1 \cup B_2 \cup B_3$ is in graph $G'$ since each one of them shares a $1$ with at least a bandpass of $M_1$.

Note that the bandpasses of $B_2$ are paired up such that each pair of the two bandpasses share a $1$ with a bandpass of $M_1$.
Assume without loss of generality that these two bandpasses of $B_2$ are formed between rows $r_i$ and $r_j$ and between rows $r_k$ and $r_\ell$,
respectively, and that the involved bandpass of $M_1$ is formed between rows $r_j$ and $r_k$ (see Figure~\ref{s2}).
That is, in the optimal row permutation $\pi^*$, rows $r_i$ and $r_j$ are adjacent, and rows $r_k$ and $r_\ell$ are adjacent;
while edge $(r_j, r_k) \in M_1$.
We remark that these four rows are distinct, and conclude that edge $(r_i, r_\ell) \notin M_1$.
The proof is simple as otherwise in the particular column a bandpass would be formed between rows $r_i$ and $r_\ell$,
making the two bandpasses of $B_2$ lose their group memberships ({\it i.e.}, they would belong to $B_3$).

\begin{lemma}
\label{lemma2}
Assume edge $(r_j, r_k) \in M_1$, and that one bandpass of $(r_j, r_k)$ shares $1$ with two bandpasses of $B_2$.
Then in graph $G$ edge $(r_j, r_k)$ is adjacent to at most four edges in the optimal row permutation $\pi^*$,
at most two of which are incident at row $r_j$,
and at most two of which are incident at row $r_k$.
\end{lemma}
\begin{proof}
The lemma is straightforward from the above discussion, and the fact that edge $(r_j, r_k)$ does not belong to the Hamiltonian path $\pi^*$.
\end{proof}

Continuing with the above discussion, assuming that edge $(r_j, r_k) \in M_1$,
and that one bandpass of $(r_j, r_k)$ shares $1$ with two bandpasses of $B_2$,
which are formed between rows $r_i$ and $r_j$ and between rows $r_k$ and $r_\ell$, respectively (see Figure~\ref{s2}).
We know that in residual graph $G'$, between rows $r_i$ and $r_\ell$, in the same column there is a bandpass
(which contributes $1$ towards the edge weight $w'(i, \ell)$).
We call bandpasses identified in this way the {\em induced} bandpasses.
From Lemma~\ref{lemma2}, edge $(r_j, r_k)$ is adjacent to at most two edges of $\pi^*$ incident at row $r_j$.
It follows that in residual graph $G'$, row $r_\ell$ can form induced bandpasses with at most four other rows.
In other words, the subgraph of $G'$ induced by the edges containing induced bandpasses, denoted as $G'_s$, is a $4$-matching in $G'$.

\begin{lemma}
\label{lemma3}
$G'_s$ is a $4$-matching in graph $G$, and its weight $w'(G'_s) \ge \frac 12 |B_2|$.
\end{lemma}
\begin{proof}
From the above discussion, $G'_s$ is a $4$-matching in residual graph $G'$.
Since the edge set of $G'$ and the edge set of $G$ are the same, disregarding edge weights, $G'_s$ is a $4$-matching of graph $G$.
The second half of the lemma can be simply argued as follows.
Since every pair of bandpasses of $B_2$ leads to an induced bandpass, all the edge weights in $G'_s$ sum up to at least $\frac 12 |B_2|$,
which is the number of bandpass pairs in $B_2$.
%
%
This finishes the proof.
\end{proof}

In $O(n^{2.5})$ time, a $4$-matching such as $G'_s$ can be decomposed into two $2$-matchings~\cite{Har69,Die05},
each of which is a collection of vertex-disjoint cycles or paths.

\begin{lemma}
\label{lem:HR}
Let ${\cal C}$ be a $2$-matching of graph $G$ such that no edge of $M_1$ is also an edge of ${\cal C}$.
Then, we can partition the edge set of ${\cal C}$ into four matchings $X_0, X_1, X_2, X_3$ 
such that $G[M_1 \cup X_j]$ is an acyclic $2$-matching for all $j \in \{0, 1, 2, 3\}$.
Moreover, the partitioning takes $O(n \alpha(n))$ time, where $\alpha(\cdot)$ is the inverse Ackerman function. 
\end{lemma}
\begin{proof}
Hassin and Rubinstein~\cite{HR00} have shown that we can compute two disjoint matchings 
$X_0$ and $X_1$ in ${\cal C}$ such that the following two conditions hold:
\begin{itemize}
\item
	Both $G[M_1\cup X_0]$ and $G[M_1\cup X_1]$ are acyclic $2$-matchings of $G$.
\item
	Each vertex of ${\cal C}$ is incident to at least one edge of $X_0 \cup X_1$. 
\end{itemize}
For convenience, let $Y$ be the set of edges in ${\cal C}$ but not in $X_0 \cup X_1$.
By the second condition, $Y$ is a matching.
Consider the graph $H = (V, M_1 \cup Y)$.
Obviously, $H$ is a collection of vertex-disjoint paths and cycles, and each cycle of $H$ contains at least two edges of $Y$.
For each cycle $C$ of $H$, we mark an arbitrary edge of $C$ that also belongs to $Y$.
Let $X_3$ be the set of marked edges, and $X_2 = Y \setminus X_3$.
Then, both $G[M_1\cup X_2]$ and $G[M_1\cup X_3]$ are acyclic $2$-matchings of $G$. 

It is not hard to see that with the famous union-find data structure~\cite{Tar75}, 
the computation of $X_0$ and $X_1$ described in \cite{HR00} can be done in $O\left(n\alpha(n)\right)$ time.
Once knowing $X_0$ and $X_1$, we can obtain $X_2$ and $X_3$ in $O(n)$ time. 
\end{proof}

In general, Lemma~\ref{lem:HR} cannot be improved by partitioning the edge set of 
${\cal C}$ into three matchings instead of four matchings. To see this, it suffices to 
consider a concrete example, where ${\cal C}$ is just a cycle of length~4 and 
$M_1$ consists of the two edges connecting nonadjacent vertices in ${\cal C}$.

Let ${\cal C}_1$ and ${\cal C}_2$ denote the two $2$-matchings constituting to the maximum weight $4$-matching ${\cal C}$ of residual graph $G'$.
Using Lemma~\ref{lem:HR} alone, ${\cal C}_1$ can be partitioned into four matchings $X_0, X_1, X_2, X_3$ and
${\cal C}_2$ can be partitioned into four matchings $Y_0, Y_1, Y_2, Y_3$,
such that $G[M_1 \cup Z_j]$ is an acyclic $2$-matching for all $Z \in \{X, Y\}$ and $j \in \{0, 1, 2, 3\}$.
The following lemma states a slightly better partition when we consider ${\cal C}_1$ and ${\cal C}_2$ simultaneously.

\begin{lemma}
\label{lemma41}
The weight of matching $M_2$ is $w'(M_2) \ge \frac 1{15} |B_2|$.
\end{lemma}
\begin{proof}
Let ${\cal C}_1$ and ${\cal C}_2$ denote the two $2$-matchings constituting to the maximum weight $4$-matching ${\cal C}$ of residual graph $G'$.
Based on the discussion in the last paragraph,
we firstly use Lemma~\ref{lem:HR} to partition the edge set of ${\cal C}_1$ into four matchings $X_0, X_1, X_2, X_3$ and
the edge set of ${\cal C}_2$ into four matchings $Y_0, Y_1, Y_2, Y_3$,
such that $G[M_1 \cup Z_j]$ is an acyclic $2$-matching for all $Z \in \{X, Y\}$ and $j \in \{0, 1, 2, 3\}$.

Note that by Lemma~\ref{lem:HR}, $X_2 \cup X_3$ is a matching and that $X_3$ contains the marked edges,
each of which, say $e = (u, v)$, is the lightest edge of the corresponding cycle, say $C$, formed in $G[M_1 \cup X_2 \cup X_3]$.
$C$ is an even cycle.
If $C$ contains at least $6$ edges, then $w'(X_3 \cap C) = w'(e) \le \frac 12 w'(X_2 \cap C)$.
The following process is to swap certain edges among $X_0, X_1, X_2, X_3$ and $Y_0, Y_1, Y_2, Y_3$ to guarantee property
\begin{itemize}
\item[(P)]
	that each of $G[M_1 \cup X_i]$ for $i = 0, 1$ and $G[M_1 \cup Y_j]$ for $j \in \{0, 1, 2, 3\}$ is an acyclic $2$-matching, and
	that $X_2 \cup X_3$ is a matching and $G[M_1 \cup X_2 \cup X_3]$ contains no length-$4$ cycles.
\end{itemize}

Let $C = (u, v, x, y)$ be a length-$4$ cycle in $G[M_1 \cup X_2 \cup X_3]$, and assume that $X_2 \cup X_3 = \{(u, v), (x, y)\}$.
Then, we call edges $(u, v)$ and $(x, y)$ a {\em problematic pair}.
Our swapping process is to {\em resolve} such problematic pairs.
We distinguish three cases.

In the first case, edges $(u, x) \notin {\cal C}_1$ and $(v, y) \notin {\cal C}_1$.

Assume the other edges of ${\cal C}_1$ incident at $u, v, x, y$ are $(u, 1)$, $(v, 2)$, $(x, 3)$, $(y, 4)$, respectively.
These four edges thus all belong to $G[M_1 \cup X_0]$ and $G[M_1 \cup X_1]$.
If at least three of them belong to $G[M_1 \cup X_0]$, then in $G[M_1 \cup X_1]$ three vertices among $u, v, x, y$ have degree $1$ and
thus they cannot be in the same connected component of $G[M_1 \cup X_1]$.
We can move (exactly) one of edges $(u, v)$ and $(x, y)$ to $X_1$, while maintaining property (P).

We examine next where exactly two of the four edges belong to $G[M_1 \cup X_0]$.
Assume without loss of generality that $(u, 1) \in G[M_1 \cup X_0]$.
If $(y, 4) \in G[M_1 \cup X_0]$, then the connected component in $G[M_1 \cup X_1]$ containing $u$ has only one edge $(u, y)$, which belongs to $M_1$.
Thus, if the other edge of ${\cal C}_1$ incident at vertex $1$ belongs to $X_1$, we can move edge $(u, 1)$ from $X_0$ to $X_2 \cup X_3$,
and move edge $(u, v)$ from $X_2 \cup X_3$ to $X_0$;
if the other edge of ${\cal C}_1$ incident at vertex $1$ does not belong to $X_1$ (and thus it must be in $X_2 \cup X_3$),
we can move edge $(u, 1)$ from $X_0$ to $X_1$, and move edge $(u, v)$ from $X_2 \cup X_3$ to $X_0$.
Either way, we maintain property (P) while resolving a problematic pair of $X_2 \cup X_3$.

If $(v, 2) \in G[M_1 \cup X_0]$, then vertices $u$ and $v$ have degree $1$ in $G[M_1 \cup X_1]$.
Thus, if the other edge of ${\cal C}_1$ incident at vertex $1$ does not belong to $X_1$, then vertex $1$ has degree $1$ in $G[M_1 \cup X_1]$ as well.
We conclude that vertices $u, v, 1$ cannot reside in the same connected component of $G[M_1 \cup X_1]$.
When $u$ and $v$ are not connected, we can move edge $(u, v)$ from $X_2 \cup X_3$ to $X_1$;
when $u$ and $1$ are not connected, we can move edge $(u, 1)$ from $X_0$ to $X_1$, and move edge $(x, y)$ from $X_2 \cup X_3$ to $X_0$.
Again, either way, we maintain property (P) while resolving a problematic pair of $X_2 \cup X_3$.
Symmetric scenarios can be argued in the same way for vertices $2, 3, 4$.
In the remaining scenario, the other edges of ${\cal C}_1$ incident at vertices $1, 2, 3, 4$ all belong to $X_0 \cup X_1$.
We then move edges $(u, 1), (v, 2), (x, 3), (y, 4)$ from $X_0 \cup X_1$ to $X_2 \cup X_3$,
and move edges $(u, v)$ ($(x, y)$, respectively) from $X_2 \cup X_2$ to $X_0$ ($X_3$, respectively).
Note that none of these four edges would form with any other edge into a problematic pair.

Lastly, if $(x, 3) \in G[M_1 \cup X_0]$, then vertices $u$ and $x$ have degree $1$ in $G[M_1 \cup X_1]$.
Thus, if the other edge of ${\cal C}_1$ incident at vertex $1$ belongs to $X_1$, then vertex $1$ has degree $1$ in $G[M_1 \cup X_2 \cup X_3]$.
We can move edge $(u, 1)$ from $X_0$ to $X_2 \cup X_3$, and move edge $(u, v)$ from $X_2 \cup X_3$ to $X_0$.
If the other edge of ${\cal C}_1$ incident at vertex $1$ does not belong to $X_1$, then vertex $1$ has degree $1$ in $G[M_1 \cup X_1]$ as well.
We conclude that vertices $u, x, 1$ cannot reside in the same connected component of $G[M_1 \cup X_1]$.
When $u$ and $1$ are not connected, we can move edge $(u, 1)$ from $X_0$ to $X_1$, and move edge $(u, v)$ from $X_2 \cup X_3$ to $X_0$.
Symmetric scenarios can be argued in the same way for vertices $2, 3, 4$.
In the remaining scenario, none of the other edges of ${\cal C}_1$ incident at vertices $1, 2, 3, 4$ belongs to $X_0 \cup X_1$,
and that vertices $u$ and $1$ ($v$ and $2$, $x$ and $3$, $y$ and $4$, respectively) are connected in $G[M_1 \cup X_1]$
($G[M_1 \cup X_0]$, $G[M_1 \cup X_1]$, $G[M_1 \cup X_0]$, respectively).
It follows that we may move edge $(u, 1)$ from $X_0$ to $X_1$, move edge $(y, 4)$ from $X_1$ to $X_0$,
and move edge $(u, v)$ from $X_2 \cup X_3$ to $X_0$, to resolve the problematic pair.

In the second case, edges $(u, x) \notin {\cal C}_1$ but $(v, y) \in {\cal C}_1$.

Assume the other edges of ${\cal C}_1$ incident at $u, x$ are $(u, 1)$, $(x, 3)$, respectively.
These two edges and edge $(v, y)$ all belong to $G[M_1 \cup X_0]$ and $G[M_1 \cup X_1]$.
Without loss of generality, assume $(v, y) \in X_1$;
it follows that vertices $v$ and $y$ have degree $1$ in $G[M_1 \cup X_0]$.
If one of edges $(u, 1)$ and $(x, 3)$ does not belong to $G[M_1 \cup X_0]$, say $(u, 1)$, then we can move $(u, v)$ from $X_2 \cup X_3$ to $X_0$,
while maintaining property (P).

If both edges $(u, 1)$ and $(x, 3)$ belong to $G[M_1 \cup X_0]$, then vertices $u$ and $x$ have degree $1$ in $G[M_1 \cup X_1]$.
When the other edge of ${\cal C}_1$ incident at vertex $1$ does not belong to $X_1$ (but $X_2 \cup X_3$),
we can move edge $(u, 1)$ from $X_0$ to $X_1$, and move edge $(u, v)$ from $X_2 \cup X_3$ to $X_0$;
the symmetric scenario can be argued in the same way for vertex $3$;
When the other edge of ${\cal C}_1$ incident at vertex $1$ and the other edge of ${\cal C}_1$ incident at vertex $3$ both belong to $X_1$,
we can move edges $(u, 1)$ and $(v, 3)$ from $X_0$ to $X_2 \cup X_3$,
move edge $(v, y)$ from $X_1$ to $X_2 \cup X_3$,
move edge $(u, v)$ from $X_2 \cup X_3$ to $X_0$,
and move edge $(x, y)$ from $X_2 \cup X_3$ to $X_1$.
Note that none of these three edges $(u, 1)$, $(v, 3)$ and $(v, y)$ would form with any other edge into a problematic pair.

In the last case, edges $(u, x) \in {\cal C}_1$ and $(v, y) \in {\cal C}_1$.

Assume without loss of generality that $(u, x) \in X_0$ and $(v, y) \in X_1$.
Since ${\cal C}_2$ do not share any edge with ${\cal C}_1$, we consider the degrees of vertices $u, v, x, y$ in $G[M_1 \cup Y_i]$ for $i = 0, 1, 2, 3$.
If in one of these four acyclic $2$-matchings, say $G[M_1 \cup Y_0]$, at least three of the four vertices have degree $1$, say $u, v, x$,
then we can move edge $(u, v)$ from ${\cal C}_1$ to $Y_0$, and thus the problematic pair of $X_2 \cup X_3$ is resolved.
In the other cases, in each $G[M_1 \cup Y_i]$ for $i = 0, 1, 2, 3$, exactly two of the four vertices have degree $1$.

Let the two edges of ${\cal C}_2$ incident at $u$ ($v, x, y$, respectively) be $(u, 1)$ and $(u, 1')$ ($(v, 2)$ and $(v, 2')$,
$(x, 3)$ and $(x, 3')$, $(y, 4)$ and $(y, 4')$, respectively).

If $(u, 1), (y, 4) \in Y_0$, then $u$ and $y$ both have degree $1$ in one of $G[M_1 \cup Y_i]$ for $i = 1, 2, 3$, say in $G[M_1 \cup Y_3]$.
It follows that if the other edge of ${\cal C}_2$ incident at vertex $1$ does not belong to $Y_3$,
then we can move edge $(u, 1)$ from $Y_0$ to $Y_3$, and move edge $(u, v)$ from ${\cal C}_1$ to $Y_0$ to resolve the problematic pair of $X_2 \cup X_3$;
or if the other edge of ${\cal C}_2$ incident at vertex $4$ does not belong to $Y_3$,
then we can move edge $(y, 4)$ from $Y_0$ to $Y_3$, and move edge $(x, y)$ from ${\cal C}_1$ to $Y_0$ to resolve the problematic pair of $X_2 \cup X_3$.
In the remaining scenario, the other edge of ${\cal C}_2$ incident at vertex $1$ (vertex $4$, respectively) belongs to $Y_3$.
Note that in either $G[M_1 \cup Y_1]$ or $G[M_1 \cup Y_2]$, vertex $u$ has degree $1$,
and we assume without loss of generality that vertex $u$ has degree $1$ in $G[M_1 \cup Y_1]$.
Note also that vertex $1$ has degree $1$ in $G[M_1 \cup Y_1]$.
If edge $(y, 4') \notin Y_1$, then vertex $y$ has degree $1$ as well, and thus
we can move edge $(u, 1)$ from $Y_0$ to $Y_1$, and move edge $(u, v)$ from ${\cal C}_1$ to $Y_0$ to resolve the problematic pair of $X_2 \cup X_3$;
if edge $(y, 4') \in Y_1$ but the other edge of ${\cal C}_2$ incident at vertex $4'$ does not belong to $Y_3$,
then we can move edge $(y, 4')$ from $Y_1$ to $Y_3$, move edge $(u, 1)$ from $Y_0$ to $Y_1$,
and move edge $(u, v)$ from ${\cal C}_1$ to $Y_0$ to resolve the problematic pair of $X_2 \cup X_3$.
Therefore, we only need to argue the scenario where the other edge of ${\cal C}_2$ incident at vertex $4'$ belongs to $Y_3$.
Symmetrically considering $Y_2$, we may assume without loss of generality that the other edge of ${\cal C}_2$ incident at vertex $1'$ belongs to $Y_3$.
Consequently, vertices $u, 1, 1'$ all have degree $1$ in $G[M_1 \cup Y_1]$, and thus $u$ and at least one of $1$ and $1'$ are not connected.
If $u$ and $1$ are not connected, we can move edge $(u, 1)$ from $Y_0$ to $Y_1$,
and move edge $(u, v)$ from ${\cal C}_1$ to $Y_0$ to resolve the problematic pair of $X_2 \cup X_3$;
if $u$ and $1'$ are not connected, we can move edge $(u, 1')$ from $Y_2$ to $Y_1$, move edge $(u, 1)$ from $Y_0$ to $Y_2$,
and move edge $(u, v)$ from ${\cal C}_1$ to $Y_0$ to resolve the problematic pair of $X_2 \cup X_3$.

If $(u, 1), (v, 2) \in Y_0$, then $u$ and $v$ both have degree $1$ in one of $G[M_1 \cup Y_i]$ for $i = 1, 2, 3$, say in $G[M_1 \cup Y_3]$.
The following discussion is very similar to the above paragraph, though slightly simpler.
Firstly, if $x$ and $y$ are not connected in $G[M_1 \cup Y_0]$ ($u$ and $v$ are not connected in $G[M_1 \cup Y_3]$, respectively),
then we can move edge $(x, y)$ ($(u, v)$, respectively) from ${\cal C}_1$ to $Y_0$ ($Y_3$, respectively)
to directly resolve the problematic pair of $X_2 \cup X_3$.
Secondly, if the other edge of ${\cal C}_2$ incident at vertex $1$ does not belong to $Y_3$,
then we can move edge $(u, 1)$ from $Y_0$ to $Y_3$, and move edge $(x, y)$ from ${\cal C}_1$ to $Y_0$ to resolve the problematic pair of $X_2 \cup X_3$;
or if the other edge of ${\cal C}_2$ incident at vertex $2$ does not belong to $Y_3$,
then we can move edge $(v, 2)$ from $Y_0$ to $Y_3$, and move edge $(x, y)$ from ${\cal C}_1$ to $Y_0$ to resolve the problematic pair of $X_2 \cup X_3$.
Symmetrically and without loss of generality that $(x, 3), (y, 4) \in Y_3$,
if either of the other edges of ${\cal C}_2$ incident at vertices $3$ and $4$ does not belong to $Y_3$, the problematic pair can be resolved.
In the remaining scenario, we assume that vertices $u$ and $x$ have degree $1$ in $G[M_1 \cup Y_1]$ (and $(v, 2'), (y, 4') \in Y_1$).
Note that vertices $1, 2, 3, 4$ all have degree $1$ in $G[M_1 \cup Y_1]$ too.
If $u$ and $x$ are not connected in $G[M_1 \cup Y_1]$, then we can swap edges of $X_0 \cup X_1$ and of $X_2 \cup X_3$,
and move edge $(u, x)$ from $X_2 \cup X_3$ to $Y_1$, to resolve the problematic pair of $X_2 \cup X_3$.
Otherwise, $u$ and $1$ should not be connected in $G[M_1 \cup Y_1]$,
and we can move edge $(u, 1)$ from $Y_0$ to $Y_1$, and move edge $(x, y)$ from $X_2 \cup X_3$ to $Y_0$,
to resolve the problematic pair of $X_2 \cup X_3$.

All the other pairs of edges occurring in ${\cal C}_2 \cap Y_0$ can be analogously discussed as in either of the above two paragraphs.
Repeatedly applying the above process to resolve the problematic pairs of $X_2 \cup X_3$, if any,
we achieve the Property (P) that each of $G[M_1 \cup X_i]$ for $i = 0, 1$ and $G[M_1 \cup Y_j]$ for $j \in \{0, 1, 2, 3\}$ is an acyclic $2$-matching,
and that $X_2 \cup X_3$ is a matching and $G[M_1 \cup X_2 \cup X_3]$ contains no length-$4$ cycles.
Subsequently, we let $X_3$ denote the set of marked edges, guaranteeing that $w'(X_3) \le \frac 12 w'(X_2)$.

It follows that at least one of $X_0, X_1, X_2, Y_0, Y_1, Y_2, Y_3$ has its weight greater than or equal to
\[
\frac 1{7.5} \left(w'({\cal C}_1) + w'({\cal C}_2)\right) \ge
\frac 1{7.5} \times \frac 12 |B_2| = \frac 1{15} |B_2|,
\]
where the last inequality follows from Lemma~\ref{lemma3} and the fact that $w'({\cal C}) \ge w'(G'_s)$.
\end{proof}

The next lemma says that Lemma~\ref{lem:HR} can be improved if the input $2$-matching is acyclic. 

\begin{lemma}\label{lem:mine}
Let ${\cal P}$ be an acyclic $2$-matching of $G$ such that no edge of $M_1$ is also an edge of ${\cal P}$.
Then, we can partition the edge set of ${\cal P}$ into three matchings $Y_0, Y_1, Y_2$
such that $G[M_1 \cup Y_j]$ is an acyclic $2$-matching for all $j \in \{0, 1, 2\}$.
Moreover, the partitioning takes $O(n \alpha(n))$ time. 
\end{lemma}
\begin{proof}
Note that ${\cal P}$ is a collection of vertex-disjoint paths.
We claim that if ${\cal P}$ has two or more connected components,
then we can connect the connected components of ${\cal P}$ into a single path by adding edges not in $M_1$ to ${\cal P}$.
To see this claim, suppose that ${\cal P}$ has two or more connected components.
Obviously, we can connect the connected components of ${\cal P}$ into a single path by adding edges to ${\cal P}$. 
Unfortunately, some edges of $M_1$ may have been added to ${\cal P}$. 
To remove edges of $M_1$ from ${\cal P}$, we start at one endpoint of ${\cal P}$ and process the edges of ${\cal P}$ in order as follows: 
\begin{itemize}
\item
	Let $s$ and $t$ be the current endpoints of ${\cal P}$, and $(u,v)$ be the current edge we want to process.
	Without loss of generality, we may assume that the removal of $(u,v)$ from ${\cal P}$ yields a path ${\cal P}_u$ from $s$ to $u$
	and another path ${\cal P}_v$ from $v$ to $t$, and further assume that the edges of ${\cal P}_u$ have been processed.
	Note that at most one of $s = u$ and $v = t$ is possible because $n \ge 3$.
	If $(u,v) \not\in M_1$, then we proceed to process the other edge incident to $v$ than $(u,v)$.
	Otherwise, $(v,s)\not\in M_1$ or $(u,t) \not\in M_1$ because $M_1$ is a matching and at most one of $s = u$ and $v = t$ is possible.
	If $(v,s) \not\in M_1$, then we modify ${\cal P}$ by deleting edge $(u,v)$ and adding edge $(v,s)$ and
	proceed to process the other edge incident to $v$ than $(v,s)$.
	On the other hand, if $(u,t) \not\in M_1$, then we modify ${\cal P}$ by deleting edge $(u,v)$ and adding edge $(u,t)$ and
	proceed to process the other edge incident to $t$ than $(u,t)$. 
\end{itemize}

By the above claim, we may assume that ${\cal P}$ is a single path ${\cal P} = (v_1, v_2, \ldots, v_{\ell+1})$,
and denote $e_j=(v_j, v_{j+1})$ for $j = 1, 2, \ldots, \ell$.

We next detail how to partition the edge set of ${\cal P}$ into three required matchings $Y_0$, $Y_1$, and $Y_2$.
Initially, we set $Y_0 = \{e_1\}$, $Y_1 = \{e_2\}$, and $Y_2 = \{e_3\}$.
Then, for $j = 4, 5, \ldots, \ell$ (in this order), we try to find a $k \in \{0,1,2\}$ such that
$Y_k \cup\{e_j\}$ is a matching and $G[M_1 \cup Y_k \cup\{e_j\}]$ is an acyclic $2$-matching of $G$. 
To explain how to find $k$, fix an integer $j \in \{4, 5, \ldots, \ell\}$. 
Let $b$ be the integer in $\{0,1,2\}$ with $e_{j-1} \in Y_b$, and $b'$ and $b''$ be the two integers in $\{0,1,2\} \setminus \{b\}$. 
If $G[M_1 \cup Y_{b'}]$ (respectively, $G[M_1 \cup Y_{b''}]$) contains no path between $v_j$ and $v_{j+1}$,
then we can set $k = b'$ (respectively, $k = b''$) and we are done.
So, we may also assume that $G[M_1 \cup Y_{b'}]$ contains a path $P'$ between $v_j$ and $v_{j+1}$ and
$G[M_1 \cup Y_{b''}]$ contains a path $P''$ between $v_j$ and $v_{j+1}$.
See Figure~\ref{fig:3part}. 

Let $v_{i'}$ (respectively, $v_{i''}$) be the neighbor of $v_j$ in $P'$ (respectively, $P''$),
and $v_{h'}$ (respectively, $v_{h''}$) be the neighbor of $v_{j+1}$ in $P'$ (respectively, $P''$).
Then, none of edges $(v_{j-1}, v_j)$, $(v_j, v_{j+1})$, and $(v_{j+1}, v_{j+2})$ can appear in $P'$ (respectively, $P''$),
because $(v_{j-1}, v_j)\in Y_b$ and neither $(v_j, v_{j+1})$ nor $(v_{j+1}, v_{j+2})$ has been processed.
So, all of $(v_j, v_{i'})$, $(v_{j+1}, v_{h'})$, $(v_j, v_{i''})$, and $(v_{j+1}, v_{h''})$ belong to $M_1$.
Thus, $i' = i''$ and $h' = h''$ because $M_1$ is a matching.
Consequently, one edge incident to $v_{i'}$ (respectively, $v_{h'}$) in ${\cal P}$ belongs to $Y_{b'}$ and the other belongs to $Y_{b''}$.
Hence, $i' < j - 1$ and $h' < j - 1$. 

Since $e_{j-1} \in Y_b$, either $e_{j-2} \in Y_{b'}$ or $e_{j-2} \in Y_{b''}$. 
We assume that $e_{j-2} \in Y_{b'}$; the case where $e_{j-2} \in Y_{b''}$ is similar. 
Since $P''$ is a path between $v_j$ and $v_{j+1}$ in $G[M_1 \cup Y_{b''}]$, $G[M_1 \cup Y_{b''}]$ contains no path between $v_j$ and $v_{j-1}$.
Thus, $G[M_1 \cup Y_{b''}\cup\{e_{j-1}\}]$ is an acyclic $2$-matching of $G$.
Hence, we move $e_{j-1}$ from $Y_{b}$ to $Y_{b''}$. 
A crucial point is that the degree of $v_{i'}$ in $G[M_1 \cup Y_b]$ is $1$. 
This is true, because $v_{i'}$ appears in both $P'$ and $P''$ and in turn cannot be incident to an edge in $Y_{b}$.
By this crucial point and the fact that $v_{i'}$ and $v_j$ belong to the same connected component in $G[M_1 \cup Y_{b}\cup\{e_{j}\}]$,
we know that $G[M_1 \cup Y_{b}\cup\{e_{j}\}]$ is an acyclic $2$-matching of $G$.
Therefore, we can set $k = b$.

\begin{figure}[ht]
\centerline{\includegraphics[angle=0,width=0.9\textwidth]{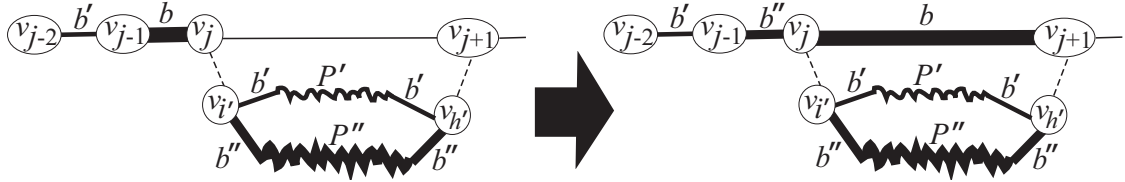}}
\caption{An illustration of moving $(v_{j-1}, v_j)$ from $Y_b$ to $Y_{b''}$ and 
adding $(v_j, v_{j+1})$ to $Y_b$, where (1)~the dashed lines indicate edges in $M_1$,  
(2)~the thin solid lines indicate edges of ${\cal P}$ that have not been processed, 
(3)~the lines labeled with $b$ (respectively, $b'$, or $b''$) indicate edges in $Y_b$ 
(respectively, $Y_{b'}$, or $Y_{b''}$), and 
(4)~the two curves may contain edges of $M_1$.} 
\label{fig:3part}
\end{figure}

Obviously, with the famous union-find data structure~\cite{Tar75},
the above partitioning of the edge set ${\cal P}$ into $Y_0, Y_1, Y_2$ can be done in $O\left(n \alpha(n)\right)$ time. 
\end{proof}

In general, Lemma~\ref{lem:mine} cannot be improved by partitioning the edge set of ${\cal P}$ into two matchings instead of three matchings.
To see this, it suffices to consider a concrete example,
where ${\cal P}$ is just a path with edges $(v_1, v_2)$, $(v_2, v_3)$, $(v_3, v_4)$ and
$M_1$ consists of edges $(v_1, v_3)$ and $(v_2, v_4)$.

\begin{lemma}
\label{lemma42}
The weight of matching $M_2$ is $w'(M_2) \ge \frac 7{27} |B_4|$.
\end{lemma}
\begin{proof}
Note that graph $G'$ contains all bandpasses of $B_4$, which is an acyclic $2$-matching.
From Lemma~\ref{lemma1}, we can compute a Hamiltonian path ${\cal P}$ in $G'$ of weight at least $\frac 79$ of the optimum,
and thus of weight at least $\frac 79 |B_4|$.
The above Lemma~\ref{lem:mine} guarantees that
\[
w'(M_2) \ge \frac 13 w'({\cal P}) \ge \frac 7{27} |B_4|.
\]
\end{proof}

\begin{theorem}
\label{theorem1}
Algorithm {\sc Approx} is an $O(n^4)$-time $\frac {426}{227}$-approximation for the Bandpass problem.
\end{theorem}
\begin{proof}
The running time of algorithm {\sc Approx} is dominated by the computing for those maximum weight $b$-matchings, for $b = 1, 2, 4$,
which can be done in $O(n^4)$ time.
Since $M_1$ is the maximum weight matching in graph $G$, from Eq.~(\ref{eq2}) we have
\begin{equation}
\label{eq5}
w(M_1) \ge \frac 12 p(\pi^*) \ge \frac 12 \left(s_2(\pi^*) + \sum_{\ell = 3}^n s_\ell(\pi^*) (\ell - 1)\right).
\end{equation}

Combining Eqs.~(\ref{eq4}) and (\ref{eq5}), we have for any real number $y \in [0, 1]$,
\begin{equation}
\label{eq6}
w(M_1) \ge y \frac 12 \left(s_2(\pi^*) + \sum_{\ell = 3}^n s_\ell(\pi^*) (\ell - 1)\right)
				 + (1-y) \left(|B_1| + \frac 12 |B_2| + \frac 23 |B_3|\right).
\end{equation}

The permutation $\pi$ produced by algorithm {\sc Approx} contains $b(\pi) \ge w(M_1) + w'(M_2)$ bandpasses, as indicated at the end of Section 2.1.
From Lemmas~\ref{lemma41} and \ref{lemma42}, we have for any real number $x \in [0, 1]$,
\begin{equation}
\label{eq7}
b(\pi) \ge w(M_1) + x \frac 1{15} |B_2| + (1-x) \frac 7{27} |B_4|.
\end{equation}

Together with Eqs.~(\ref{eq3}) and (\ref{eq6}), the above Eq.~(\ref{eq7}) becomes,
\begin{eqnarray}
\label{eq8}
b(\pi)	&\ge	&w(M_1) + x \frac 1{15} |B_2| + (1-x) \frac 7{27} |B_4|\nonumber\\
		&\ge	&y \frac 12 \left(s_2(\pi^*) + \sum_{\ell = 3}^n s_\ell(\pi^*) (\ell - 1)\right)\nonumber\\
		&		&	+ (1-y) \left(|B_1| + \frac 12 |B_2| + \frac 23 |B_3|\right)
					+ x \frac 1{15} |B_2| + (1-x) \frac 7{27} |B_4|\nonumber\\
		&=		&\frac y2 \left(s_2(\pi^*) + \sum_{\ell = 3}^n s_\ell(\pi^*) (\ell - 1)\right)\nonumber\\
		&		& + (1-y) |B_1| + \left(\frac {1-y}2 + \frac x{15}\right) |B_2| + \frac {2(1-y)}3 |B_3| + \frac {7(1-x)}{27} |B_4|\nonumber\\
		&\ge	&\frac {57}{142} \left(s_2(\pi^*) + \sum_{\ell = 3}^n s_\ell(\pi^*) (\ell - 1)\right)
				    + \frac {14}{213} |B_1| + \frac {28}{213} s_2(\pi^*),
\end{eqnarray}
where the last inequality is achieved by setting $x = \frac {35}{71}$ and $y = \frac {57}{71}$.
Note that for all $\ell \ge 3$, $(\ell - 1) \ge \frac 32 \lfloor\frac {\ell}2\rfloor$.
It then follows from Eqs.~(\ref{eq8}) and (\ref{eq1}) that
\begin{equation}
\label{eq9}
b(\pi) \ge \frac {227}{426} \left(s_2(\pi^*) + \frac {171}{227} \times \frac 32 \sum_{\ell = 3}^n s_\ell(\pi^*) \left\lfloor\frac {\ell}2\right\rfloor\right)
	\ge \frac {227}{426} b(\pi^*).
\end{equation}

That is, the worst-case performance ratio of algorithm {\sc Approx} is at most $\frac {426}{227}$.
\end{proof}

\section{Conclusions and future work}
In this paper, we presented a $\frac {426}{227}$-approximation algorithm for the Bandpass problem ($\frac {426}{227} \approx 1.8767$),
which improves the first non-trivial approximation ($\approx 1.8948$) and a subsequent approximation ($\approx 1.8805$).
Our algorithm is based on maximum weight $b$-matchings, for $b = 1, 2, 4$, similar to tackling the closely related Max-TSP.
The intrinsic structural property we proved for the optimal row permutation and the maximum weight matching is fundamental,
without which no better lower bound on the optimum can be built.
The schemes we developed to partition a $b$-matching, for $b = 2$ and $b = 4$, into desired matchings are also interesting,
and could potentially be further improved.

For the Max-TSP, Serdyukov presented a $\frac 43$-approximation algorithm based on
the maximum weight {\em assignment} (or called {\em cycle cover}) and the maximum weight matching~\cite{Ser84},
which has been improved to the currently best $\frac 97$-approximation algorithm in Lemma~\ref{lemma1}.
We believe that the Bandpass problem can be better approximated by
either improving the two key $b$-matching partitioning schemes presented in this paper,
or introducing new structural properties and/or new techniques;
yet we also believe that there will be a gap from $\frac 97$, due to the ``dynamic'' edge weights.

On the other hand, Hassin and Rubinstein gave a randomized approximation algorithm for the Max-TSP
with expected performance ratio $\frac {33}{25} ( = 1.32)$~\cite{HR00} (which was subsequently de-randomized in \cite{COW05}).
It would be interesting to design a randomized approximation for the Bandpass problem too, with a better than $1.8767$ expected performance ratio.

\section*{Acknowledgement}
Weitian Tong, Randy Goebel, and Guohui Lin are supported in part by NSERC.


\begin{thebibliography}{10}

\bibitem{Ans87}
R.~P. Anstee.
\newblock A polynomial algorithm for $b$-matching: An alternative approach.
\newblock {\em Information Processing Letters}, 24:153--157, 1987.

\bibitem{AH98}
E.~M. Arkin and R.~Hassin.
\newblock On local search for weighted packing problems.
\newblock {\em Mathematics of Operations Research}, 23:640--648, 1998.

\bibitem{BBN09}
D.~A. Babayev, G.~I. Bell, and U.~G. Nuriyev.
\newblock The bandpass problem: combinatorial optimization and library of
  problems.
\newblock {\em Journal of Combinatorial Optimization}, 18:151--172, 2009.

\bibitem{BB04}
G.~I. Bell and D.~A. Babayev.
\newblock Bandpass problem.
\newblock In {\em Annual INFORMS meeting}, 2004.
\newblock October 2004, Denver, CO, USA.

\bibitem{CH99}
B.~Chandra and M.~M. Halld\'{o}rsson.
\newblock Greedy local improvement and weighted set packing approximation.
\newblock In {\em ACM-SIAM Proceedings of the Tenth Annual Symposium on
  Discrete Algorithms (SODA'99)}, pages 169--176, 1999.

\bibitem{COW05}
Z.-Z. Chen, Y.~Okamoto, and L.~Wang.
\newblock Improved deterministic approximation algorithms for {Max} {TSP}.
\newblock {\em Information Processing Letters}, 95:333--342, 2005.

\bibitem{CW12}
Z.-Z. Chen and L.~Wang.
\newblock An improved approximation algorithm for the bandpass-2 problem.
\newblock In {\em Proceedings of the 6th Annual International Conference on
  Combinatorial Optimization and Applications (COCOA 2012)}, volume 7402 of
  {\em LNCS}, pages 185--196, 2012.

\bibitem{Die05}
R.~Diestel.
\newblock {\em Graph Theory}.
\newblock Springer, 3rd edition, 2005.

\bibitem{Gab83}
H.~Gabow.
\newblock An efficient reduction technique for degree-constrained subgraph and
  bidirected network flow problems.
\newblock In {\em Proceedings of the 15th Annual ACM Symposium on Theory of
  Computing (STOC'83)}, pages 448--456, 1983.

\bibitem{GJ79}
M.~R. Garey and D.~S. Johnson.
\newblock {\em Computers and Intractability: A Guide to the Theory of
  NP-completeness}.
\newblock W. H. Freeman and Company, San Francisco, 1979.

\bibitem{Har69}
F.~Harary.
\newblock {\em Graph Theory}.
\newblock Addison-Wesley, 1969.

\bibitem{HR00}
R.~Hassin and S.~Rubinstein.
\newblock Better approximations for {Max} {TSP}.
\newblock {\em Information Processing Letters}, 75:181--186, 2000.

\bibitem{Lin11}
G.~Lin.
\newblock On the {Bandpass} problem.
\newblock {\em Journal of Combinatorial Optimization}, 22:71--77, 2011.

\bibitem{MP95}
D.~L. Miller and J.~F. Pekny.
\newblock A staged primal-dual algorithm for perfect $b$-matching with edge
  capacities.
\newblock {\em ORSA Journal on Computing}, 7:298--320, 1995.

\bibitem{PMM09}
K.~E. Paluch, M.~Mucha, and A.~Madry.
\newblock A 7/9 - approximation algorithm for the maximum traveling salesman
  problem.
\newblock In {\em Proceedings of the 12th International Workshop on APPROX and
  the 13th International Workshop on RANDOM}, volume 5687 of {\em LNCS}, pages
  298--311, 2009.

\bibitem{Ser84}
A.~I. Serdyukov.
\newblock An algorithms for with an estimate for the traveling salesman problem
  of the maximum.
\newblock {\em Upravlyaemye Sistemy}, 25:80--86, 1984.

\bibitem{Tar75}
R.~E. Tarjan.
\newblock Efficiency of a good but not linear set union algorithm.
\newblock {\em Journal of the ACM}, 22:215--225, 1975.

\bibitem{TGD12}
W.~Tong, R.~Goebel, W.~Ding, and G.~Lin.
\newblock An improved approximation algorithm for the bandpass problem.
\newblock In {\em Proceedings of the Joint Conference of the Sixth
  International Frontiers of Algorithmics Workshop and the Eighth International
  Conference on Algorithmic Aspects of Information and Management (FAW-AAIM
  2012)}, LNCS 7285, pages 351--358, 2012.

\end{thebibliography}

\end{document}